\documentclass[nohyperref]{article}
\usepackage{microtype}
\usepackage{graphicx}
\usepackage{subfigure}
\usepackage{booktabs} 

\usepackage{hyperref}

\usepackage[accepted]{icml2023}

\usepackage{amsmath}
\usepackage{amssymb}
\usepackage{mathtools}
\usepackage{amsthm}

\usepackage[capitalize,noabbrev]{cleveref}

\theoremstyle{plain}
\newtheorem{theorem}{Theorem}[section]

\newtheorem{lemma}[theorem]{Lemma}

\theoremstyle{definition}

\theoremstyle{remark}

\usepackage[textsize=tiny]{todonotes}
\usepackage{amsmath}
\usepackage{graphics}
\usepackage{epsfig}
\usepackage{multirow}
\usepackage{xspace}
\usepackage{subfigure}
\usepackage{fancyhdr}
\usepackage{microtype}
\usepackage{color}
\usepackage{colortbl}
\usepackage{xcolor}
\definecolor{mygray}{gray}{.85}
\definecolor{myyellow}{RGB}{204,102,0}
\definecolor{myred}{RGB}{204,0,102}
\definecolor{mypurple}{RGB}{102,0,204}
\usepackage{amsmath}
\usepackage{graphics}
\usepackage{epsfig}
\usepackage{multirow}
\usepackage{xspace}
\usepackage{subfigure}
\usepackage{fancyhdr}
\usepackage{amssymb}
\usepackage{amsthm}
\usepackage{mathrsfs}
\usepackage{booktabs}
\usepackage{arydshln}
\usepackage{enumitem}
\DeclareMathOperator*{\argmax}{argmax} 

\usepackage{algorithm}
\usepackage{algorithmic}
\renewcommand{\algorithmicrequire}{\textbf{Input:}} 
\renewcommand{\algorithmicensure}{\textbf{Output:}} 
\newcommand{\model}{IsEM-Pro\xspace}

\icmltitlerunning{Importance Weighted Expectation-Maximization for Protein Sequence Design}

\begin{document}

\twocolumn[
\icmltitle{Importance Weighted Expectation-Maximization for Protein Sequence Design}




\begin{icmlauthorlist}
\icmlauthor{Zhenqiao Song}{comp}
\icmlauthor{Lei Li}{comp}
\end{icmlauthorlist}

\icmlaffiliation{comp}{Department of Computer Science, University of California, Santa Barbara, California, the United States}

\icmlcorrespondingauthor{Zhenqiao Song}{zhenqiao@ucsb.edu}
\icmlcorrespondingauthor{Lei Li}{leili@cs.ucsb.edu}

\icmlkeywords{Machine Learning, ICML}

\vskip 0.3in
]



\printAffiliationsAndNotice{}  

\begin{abstract}
Designing protein sequences with desired biological function is crucial in biology and chemistry.
Recent machine learning methods use a surrogate sequence-function model to replace the expensive wet-lab validation.
How can we efficiently generate diverse and novel protein sequences with high fitness?
In this paper, we propose \model, an approach to generate protein sequences towards a given fitness criterion.
At its core, \model is a latent generative model, augmented by combinatorial structure features from a separately learned Markov random fields (MRFs). We develop an Monte Carlo Expectation-Maximization method (MCEM) to learn the model. During inference, sampling from its latent space enhances diversity while its MRFs features guide the exploration in high fitness regions.
Experiments on eight protein sequence design tasks show that our \model outperforms the previous best methods by at least 55\% on average fitness score and generates more diverse and novel protein sequences. The code is available at \url{https://github.com/JocelynSong/IsEM-Pro.git}

\end{abstract}

\section{Introduction}
\label{introduction}
Protein engineering aims to discover protein variants with desired biological function, such as fluorescence intensity \cite{biswas2021low}, enzyme activity \cite{fox2007improving}, and therapeutic efficacy \cite{lagasse2017recent}.
Protein sequences embody their function through spontaneous folding of amino-acid sequences into three dimensional structures  \cite{go1983theoretical,chothia1984principles,starr2017exploring}.
The mapping from protein sequence to functional property forms a protein fitness landscape that characterizes the protein functional levels, such as the capability to catalyze reaction or bind a specific ligand \cite{romero2009exploring,ren2022proximal}.
Traditional approaches to design a specific protein with desired fitness objective involve obtaining protein variants by random mutagenesis \cite{labrou2010random} or recombination in laboratory experiments \cite{ma2003production}.
These variants are screened and selected in wet-lab experiments \cite{arnold1998design} as illustrated in Figure \ref{figure1}.

\begin{figure}[t]
\begin{minipage}[t]{0.5\linewidth}
\centering
\includegraphics[width=4.1cm]{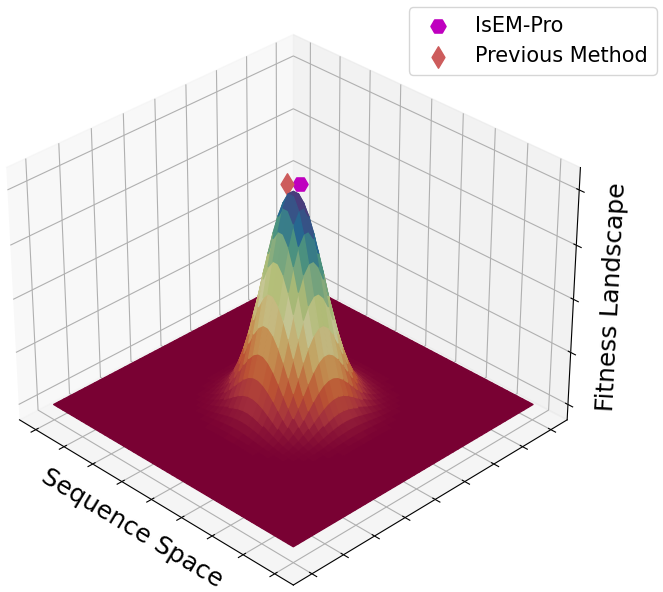}
\centerline{(a) Fujiyama landscape.}
\end{minipage}%
\begin{minipage}[t]{0.5\linewidth}
\centering
\includegraphics[width=4.1cm]{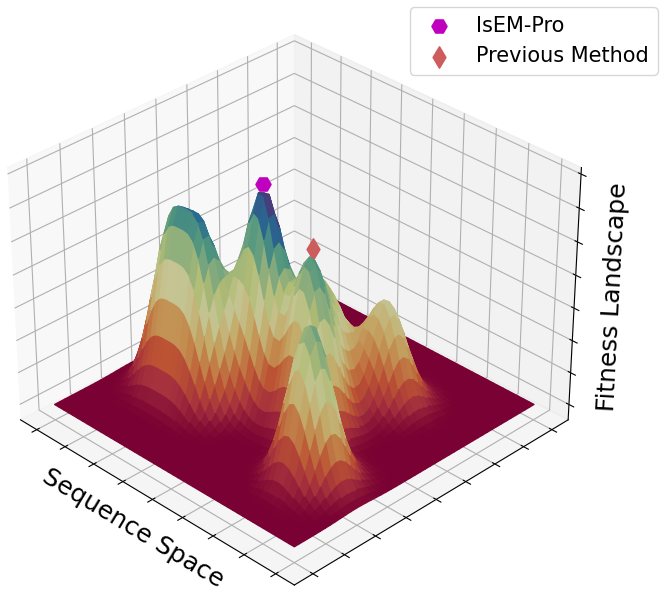}
\centerline{(b) Badlands landscape.}
\end{minipage}
\vspace{-1em}
\caption{Protein Fitness Landscape: distribution of a functional property for proteins).
Protein may exhibit single-peaked fitness landscape (Fujiyama landscape (a)) or multi-peaked landscape (Badlands landscape (b)) \cite{kauffman1989nk}.
In Fujiyama landscape, any method could perform well. However, for the rougher Badlands landscape, previous methods get trapped in a worse local optima while our proposed \model can climb much closer to the global optima through the iterative sampling in the latent space.} 
 \label{figure2}
\end{figure}

\begin{figure*}[ht]
  \centering
  \includegraphics[width=13.0cm]{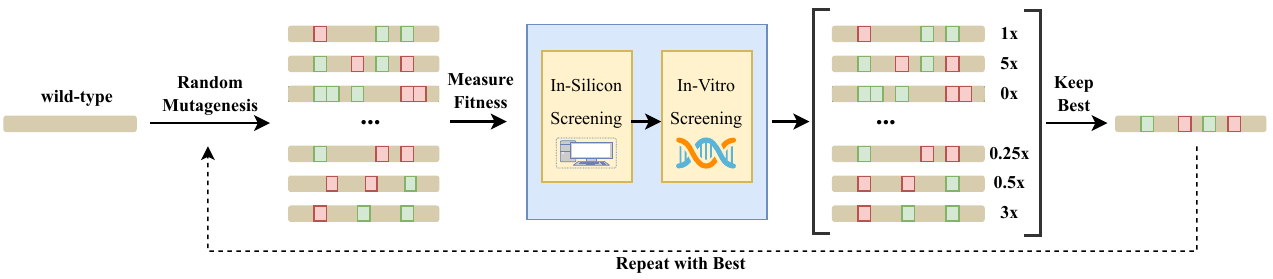}
  \vspace{-1em}
  \caption{Workflow of traditional protein sequence design. We aim to accelerate this process by directly generating desirable sequences.}
  \label{figure1}
\end{figure*}

However, these approaches requires iterative cycles of random mutagenesis and wet-lab validation, which are both money-consuming and time-intensive. 
Recent machine learning methods attempt to build a surrogate model of protein fitness landscape to accelerate expensive wet-lab screening \cite{luo2021ecnet,meier2021language}.
How can we efficiently discover satisfactory proteins over the exponentially large discrete space?
Ideal protein molecules should be novel, diverse and exhibit high fitness.
On one hand, designing novel and diverse protein sequences can uncover new functions and also lead to functional diversification \cite{singh2016plant}.
On the other hand, in addition to the evolutionary pressure for a protein to conserve specific positions of its sequence for function, the diversification of protein sequences can avoid undesired inter-domain association such as misfolding \cite{wright2005importance}.


 In this paper, we propose an \textbf{I}mportance \textbf{s}ampling based \textbf{E}xpectation-\textbf{M}aximization (EM) method to efficiently design novel, diverse and desirable \textbf{Pro}tein sequences (\model). 
Specifically, we introduce a latent variable in the generative model to capture the inter-dependencies in protein sequences. Sampling in latent space leads to more diverse candidates and can escape from locally optimal fitness regions.
Instead of using standard variational inference models such as variational auto-encoder (VAE) \cite{kingma2013auto}, we leverage importance sampling inside the EM algorithm to learn the latent generative model.
As illustrated in Figure \ref{figure2}, our approach can navigate through multiple local optimum, and yield better overall performance.
We further incorporate combinatorial structure of amino acids in protein sequences using Markov random fields (MRFs).
It guides the model towards higher fitness landscape, leading to faster uphill path to desired proteins.

We carry out extensive experiments on eight protein sequence design tasks and compare the proposed method with previous strong baselines. The contribution are as follows: 
\begin{itemize}[nosep]
\item We propose a structure-enhanced latent generative model for protein sequence design.
\item We develop an efficient method to learn the proposed generative model, based on importance sampling inside the EM algorithm.
\item Experiments on eight protein datasets with different objectives demonstrate that our \model generates protein sequences with at least 55\% higher average fitness score and higher diversity and novelty than previous best methods. Further analyse show that the protein sequences designed by our model can fold stably, giving empirical evidence that our proposed \model has the ability to generate real proteins.
\end{itemize}



\section{Background}
\label{background}

\subsection{Protein Sequence Design upon Wild-Type }
The protein sequence design problem is to search for a sequence with desired property in the sequence space $\mathcal{V}^L$, where $\mathcal{V}$ denotes the vocabulary of amino acids and L denotes the desired sequence length. 
The target is to find a protein sequence with highest fitness given by a protein fitness function $f: \mathcal{V}^L \rightarrow \mathbb{R}$, which can be measured through wet-lab experiments. 
Wild-type refers to protein occurring in nature.
Evolutionary search based methods are widely used \cite{bloom2009light,arnold2018directed,angermueller2019model,ren2022proximal}.
They use wild-type sequence as starting point during iterative search.
In this paper, we do not focus on the modification upon the wild-type sequence, but aim to efficiently generate novel and diverse sequences with improved protein functional properties.

\subsection{Monte Carlo Expectation-Maximization}
We will develop our method based on Monte Carlo expectation-maximization (MCEM)~\cite{bishop2006pattern}.
A latent generative model assumes data x (e.g. a protein sequence) is generated from a latent variable z.
To learn this latent model, the optimization procedure for maximizing the log marginal likelihood is to alternate between expectation step (E-step) and maximization step (M-step). 
EM directly targets the log marginal likelihood of an observation x by involving a variational distribution $q_{\phi} (z)$:
\begin{equation}
\begin{split}
\log p_{\theta}(x) &= E_{q_{\phi}(z)} [\log p_{\theta}(x, z) - \log q_{\phi}(z)] \\&+ D_{KL}(q_{\phi}(z)||p_{\theta}(z|x))
\end{split}
\label{equation_variational_em}
\end{equation}
where $p_{\theta}(z|x)$ is the true posterior distribution and $p_{\theta}(x, z)=p_{\theta}(x|z)p_{\theta} (z)$ is the joint distribution, composed of the conditional likelihood $p_{\theta}(x|z)$ and the prior $p_{\theta}(z)$.
In MCEM, E-step samples a set of z from $q_{\phi}(z)$ to estimate the expectation using Monte Carlo method, and then M-step fits the model parameters $\theta$ by maximizing the Monte Carlo estimation \cite{wei1990monte}.
It can be proved that this process will never decease the log marginal likelihood \cite{bishop2006pattern}.


\section{Proposed Method: \model}
\label{methods}
In this section, we describe our method in detail.
We will first present the probabilistic model and its learning algorithm.
To make the learning more efficient, we describe how to uncover and use the potential constraints conveyed in the protein sequences.

\subsection{Problem Formulation} 
Our goal is to search over a space of discrete protein sequences $\mathcal{V}^{L}$ -- $\mathcal{V}$ consists of 20 amino acids and $L$ is the sequence length -- for sequence $x \in \mathcal{V}^{L}$ that maximizes a given fitness function $f: \mathcal{V}^L \rightarrow \mathbb{R}$.  
Let the fitness value $y=f(x)$, given a predefined threshold $\lambda$, we define a conditional likelihood function:
\begin{equation}
P(\mathcal{S}|x) = \left\{
\begin{aligned}
1 &,  & {f(x) \ge \lambda,} \\
0 &, & {\text{otherwise}}
\end{aligned}
\right.
\label{equation_labmbda}
\end{equation}
where $\mathcal{S}$ represents the event that the fitness of x is ideal ($y\ge \lambda$).
We also assume a class of generative models $P_{\theta}(x)$ that can be trained to model the raw protein sequences and it will be kept fixed afterwards.
Since the search space is exponentially large (O($20^L$)), random search would be time-intensive.
Following \citet{brookes2019conditioning}, we formulate the protein design problem as generating satisfactory sequences from the posterior distribution $P_{\theta}(x|\mathcal{S})$:
\begin{equation}
P_{\theta}(x|\mathcal{S})=\frac{P_{\theta}(x)P(\mathcal{S}|x)}{P_{\theta}(\mathcal{S})}
\label{eq:posterior_fitted_protein}
\end{equation}
where $P_{\theta}(\mathcal{S})=\int_x P_{\theta}(x)P(\mathcal{S}|x) dx$ is a normalization constant which does not rely on x.
Protein sequences generated from $P_{\theta}(x|\mathcal{S})$ are not only more likely to be real proteins, but also have higher functional scores (i.e., fitness). The higher the $\lambda$ is, the higher fitness the discovered protein sequences have.

\subsection{Probabilistic Model}
Directly generating satisfactory sequences from the posterior distribution $P_{\theta}(x|\mathcal{S})$ is highly efficient compared with randomly search over the exponentially discrete space.
However, realizing this idea is difficult as computing $P_{\theta}(\mathcal{S})=\int_x P(\mathcal{S}|x)P_{\theta}(x)dx$ needs an integration over all possible x, which is intractable.
Instead, we propose to learn a variational distribution $Q_{\phi}(x)$ with learnable parameter $\phi$ to approximate satisfactory proteins $P_{\theta}(x|\mathcal{S})$.
Following \citet{brookes2019conditioning}, to find the optimal $\phi$ of the proposal distribution, we choose to minimize the KL divergence between the posterior distribution $P_{\theta}(x|\mathcal{S})$ and the variational distribution $Q_{\phi}(x)$:
\begin{equation}
\begin{split}
\phi^* &= \argmax_{\phi} -D_{KL}(P_{\theta}(x|\mathcal{S})||Q_{\phi}(x)) \\
&= \argmax_{\phi} E_{P_{\theta}(x|\mathcal{S})} \log Q_{\phi}(x) + \mathcal{H(P_{\theta})}
\end{split}
\label{equation4}
\end{equation}
where $\mathcal{H(P_{\theta})}=-E_{P_{\theta}(x|\mathcal{S})} \log P_{\theta}(x|\mathcal{S})$ is the entropy of $P_{\theta}(x|\mathcal{S})$ and can be dropped because it does not matter $\phi$. 

\begin{figure}
  \centering
  \includegraphics[width=8.0cm]{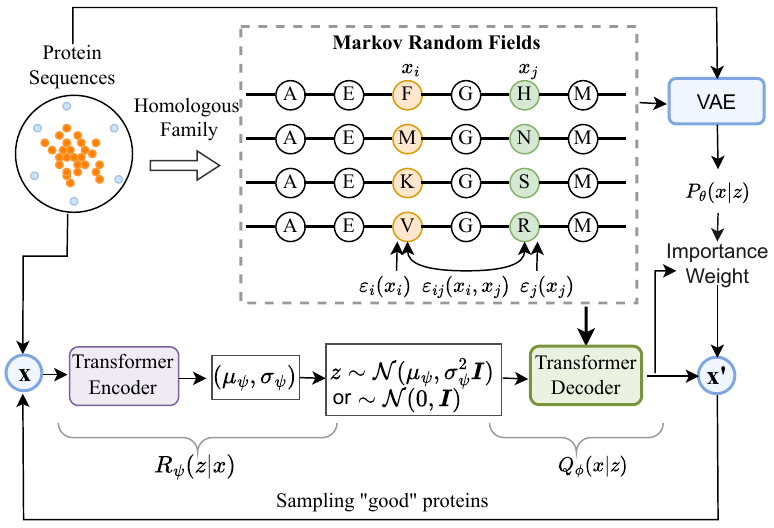}
  \caption{Overall architecture of the proposed \model. The upper half illustrates the Markov random fields which learns the combinatorial structure of amino acids in protein sequences from the same family. The bottom half shows the combinatorial structure feature augmented probabilistic model. Red lines show the calculation of importance weight.}
  \vspace{-1em}
  \label{figure_model}
\end{figure}

Diversity is a key consideration in our protein design procedure, which not only satisfies the diverse nature of species, but also can reduce undesired inter-domain misfolding \cite{wright2005importance}.
In order to promote the diversity of the designed protein sequences, we introduce a latent variable z into our model to capture the high-order dependencies among amino acids in protein sequences. 
We assume the joint approximate distribution $Q_\phi(x, z) = Q_0(z) Q_\phi(x|z)$. 
Thus our final goal is to maximize the expected log-likelihood  $\log Q_{\phi}(x)$ of the satisfactory proteins with respect to the ideal posterior distribution $P_{\theta}(x|\mathcal{S})$. 
By using the EM objective from Eq.\eqref{equation_variational_em}, we derive  
\begin{align}
\mathcal{L} = E_{P_{\theta}(x|\mathcal{S})} &\log Q_{\phi}(x) \notag \\
= E_{P_{\theta}(x|\mathcal{S})}
&\bigl\{E_{R_{\psi}(z|x)} [\log Q_{\phi}(x, z) - \log R_{\psi}(z|x)] \notag \\ 
&+ D_{KL}(R_{\psi}(z|x)||Q_{\phi}(z|x))\bigr\} \notag \\
= E_{P_{\theta}(x|\mathcal{S})}& \mathcal{F}(R_{\psi}(z|x), \phi)
\label{equation_likelihodd}
\end{align}
where $R_{\psi}(z|x)$ is another variational distribution with its parameter $\psi$ to approximate the intractable posterior $Q_\phi(z|x)$, and $\mathcal{F}(R_{\psi}(z|x), \phi)=E_{R_{\psi}(z|x)} [\log Q_{\phi}(x, z) - \log R_{\psi}(z|x)]+ D_{KL}(R_{\psi}(z|x)||Q_{\phi}(z|x))$. Here $Q_0(z)$ is implemented as a standard normal distribution, and $Q_{\phi}(x|z)$ is a Transformer decoder \cite{vaswani2017attention} with $z$ as the first embedding input, which will be augmented by the combinatorial structure features (Sec. \ref{subsection_mrfs}). 
We implement $R_{\psi}(z|x)$ as a normal distribution $\mathcal{N}(\mu_{\psi}, \sigma^2_{\psi}\bold{I})$ with $\mu_{\psi}, \sigma_{\psi}=\mathrm{Transformer}_\psi(x)$.
The overall architecture is illustrated in Figure~\ref{figure_model}.



\subsection{Importance Weighted EM}
To maximize the objective defined in Eq.\eqref{equation_likelihodd}, we plan to learn the proposal distribution $Q_{\phi}(x, z)$ through Monte Carlo EM, because the sampling procedure and iterative optimization process can lead to a better estimate (Figure \ref{figure2} (b)), resulting in novel and diverse proteins with higher fitness.

Since Eq.\eqref{equation_likelihodd} involves  the  expectation over the intractable distribution $P_{\theta}(x|\mathcal{S})$ and the KL divergence of the intractable posterior distribution $Q_\phi(z|x)$, 
we use importance sampling to approximate its variational lower bound. 
We assume the original protein sequence is also generated from a same latent variable $z$, i.e. $P_{\theta}(x,z)=P_{\theta}(x|z)P_0(z)$, where $P_0(z)$ is a standard normal distribution and $P_{\theta}(x|z)$ is a Transformer with $z$ as the initial input embedding. 
We also assume the latent variable z in $P_{\theta}(x, z)$ and $Q_{\phi}(x, z)$ are defined on the same latent space.
In practice, we learn $P_{\theta}(x|z)$ as the decoder in a pre-trained variational auto-encoder on raw protein sequences. 
We assume $\mathcal{S}$ is only conditioned on x, leading to  $P_{\theta}(x, z| \mathcal{S})=\frac{P_{\theta}(x,z)P(\mathcal{S}|x)}{P_{\theta}(\mathcal{S})}$.
Then the final training objective can be derived as:
\begin{align}
\mathcal{L} &=  E_{P_{\theta}(x|\mathcal{S})} \mathcal{F}(R_{\psi}(z'|x), \phi) = E_{P_{\theta}(x, z|\mathcal{S})} \mathcal{F}(R_{\psi}(z'|x), \phi) \notag \\
 &= E_{{Q_{\phi}(x,z)}}\frac{P_{\theta}(x,z|\mathcal{S})}{Q_{\phi}(x,z)}\mathcal{F}(R_{\psi}(z'|x), \phi) \notag \\
&\ge \frac{1}{\mathcal{C}}E_{{Q_{\phi}(x,z)}}\frac{P_{\theta}(x|z)}{Q_{\phi}(x|z)}P(\mathcal{S}|x) \Bigl( E_{R_{\psi}(z'|x)} \log Q_{\phi}(x|z') \notag \\
& \qquad\qquad- D_{KL}(R_{\psi}(z'|x)||Q_0(z'))\Bigr)  = \tilde{\mathcal{L}}
\label{equation8}
\end{align}
Here we use the importance sampling fundamental identity to derive the  equation~\cite{robert2004monte}.
$\mathcal{C} = P_{\theta}(\mathcal{S})$ is a constant which does not rely on $\phi$ and $\psi$.
The last inequality adopts the evidence lower bound (ELBO)~\cite{bishop2006pattern}.
We use importance sampling based EM to approximate the above objective \cite{robert2004monte} with joint samples $(x_n, z_n) \sim Q_{\phi}(x, z)$.
Specifically, at each iteration, we perform, \\
\textbf{E-step}:
\begin{align}
&\text{Sample}~~ x_n \sim P_{data}, ~~ z_n\sim R_{\psi}(z|x_n) \notag \\
&~~~~~~\text{ and } z_n\sim Q_0(z), ~~ x_n\sim Q_{\phi}(x|z_n), ~~ z'_n \sim R_{\psi}(z|x_n) \notag \\
&\tilde{\mathcal{L}_t} = \sum_{n=1}^{N} w'(x_n, z_n) \Bigl(\log Q_\phi (x_n | z'_n) \notag \\
& \qquad\qquad - D_{KL}(R_{\psi}(z'|x)||Q_0(z'))\Bigr)
\label{eq:e-step-init}
\end{align}
where $w(x_n, z_n)=\frac{P_{\theta}(x_n|z_n)}{Q_{\phi^{(t)}}(x_n|z_n)}P(\mathcal{S}|x_n)$ is the unnormalized importance weight, $w'(x_n, z_n)=\frac{ w(x_n, z_n)}{\sum_{n=1}^N w(x_n, z_n)}$ is the normalized one, and $N$ is the sample size.  Since $R_\psi(z|x)$ is assumed to be a normal distribution with its mean $\mu_\psi$ and variance $\sigma_\psi^2 \bold{I}$ calculated from a Transformer encoder $\mu_{\psi}, \sigma_{\psi}=\mathrm{Transformer}_\psi(x_n)$, the KL term can be calculated in closed form (Eq. \eqref{equation: special_kl}). \\
\textbf{M-step}: 
\begin{align}
\label{eq:m-step-init}
&\psi^{(t+1)}, \phi^{(t+1)} = \argmax_{\psi, \phi} \tilde{\mathcal{L}_t} \\
&\psi^{(t+1)} = \psi^{(t)} + r * \nabla_{\psi} \tilde{\mathcal{L}_t}, \phi^{(t+1)} = \phi^{(t)} + r * \nabla_{\phi} \tilde{\mathcal{L}_t} 
\notag
\end{align}
In E-step, we use two techniques to generate samples for $z$ and $x$  -- using the real protein sequences $x$ with generated $z$ from $R_\psi(z|x)$ and using $z$ from a standard normal distribution with generated proteins from $Q_\phi(x|z)$.
In M-step, we optimize the KL regularized data generation log likelihood with both high-fitness real and synthetic proteins. 
The procedure can be viewed as self-training with augmented synthetic protein data, which differs from prior approaches such as CbAS~\cite{brookes2019conditioning}.  

\subsection{Guiding Model Climbing through Combinatorial Structure}
\label{subsection_mrfs}
As shown in previous work, the combinatorial structure of amino acids in protein sequences can be learned from a generative graphical model Markov random fields (MRFs) fitted on the sequences from the same family \cite{hopf2017mutation,luo2021ecnet}.
These structure constraints are the results of the evolutionary process under natural selection and may reveal clues on which amino-acid combinations are more favorable than others. 
Thus we incorporate these features into our model to guide it towards higher fitness landscape to faster find desired protein sequences.

Given a protein sequence $x=(x_1, x_2, .., x_M)$ with $M$ amino acids, the generative model generates it with likelihood $P_{\epsilon}(x)=\frac{\exp(\Upsilon(x))}{Z}$ where $Z=\int_x \exp(\Upsilon(x)) dx$ is a normalization constant and $\Upsilon(x)$ is the corresponding energy function, which is defined as the sum of all pairwise constraints and single-site constraints as follows:
\begin{equation}
\Upsilon(x) = \sum_{i=1}^M \varepsilon_i(x_i) + \sum_{i=1}^M\sum_{j=1, j\ne i}^M \varepsilon_{ij}(x_i, x_j)
\end{equation}
where $\varepsilon_i(x_i)$ denotes the single-site constraint of $x_i$ at position i and $\varepsilon_{ij}(x_i, x_j)$ denotes the pairwise constraint of $x_i$ and $x_j$ at position i, j.
The above graphical model is illustrated in the upper half of Figure~\ref{figure_model}.

We train the model on protein sequences from the same family following CCMpred \cite{seemayer2014ccmpred} using a pseudo-likelihood $\hat{P}_{\epsilon}(x)$ (provided in Appendix \ref{pseudo_likelihood}) combined with $L_2$ regularization to make the  learning of $P_{\epsilon}(x)$ easier. 
But different from them, we additionally add $L_1$ regularization to the training objective to make the graph sparse, of which the regularization coefficients are set to the same values as the $L_2$ regularization:
\begin{equation}
\begin{split}
\max L_{\epsilon}&=\sum_x \log \hat{P}_{\epsilon}(x) - L_1(\varepsilon) - L_2(\varepsilon) \\
L_1(\varepsilon) &= \alpha_{\text{single}} \sum_{i=1}^M ||\varepsilon_i||_1^1 + \alpha_{\text{pair}}\sum_{i,j=1,i\ne j}||\varepsilon_{ij}||_1^1 \\
L_2(\varepsilon) &= \alpha_{\text{single}} \sum_{i=1}^M ||\varepsilon_i||_2^2 + \alpha_{\text{pair}}\sum_{i,j=1,i\ne j}||\varepsilon_{ij}||_2^2 \\
\end{split}
\end{equation}
where $x$ denotes protein sequences from the same family, $\varepsilon_i=[\varepsilon_i(a_1), \varepsilon_i(a_2), ..., \varepsilon_i(a_{20})]$ is the vector of the single-site constraints of the 20 amino acids at position i, and $\varepsilon_{ij}=[\varepsilon_{ij}(a_1, a_2), \varepsilon_{ij}(a_1, a_3), ..., \varepsilon_{ij}(a_M, a_{M-1})]$ is the vector of all possible pairwise constraints at position i, j.


After training the MRFs, we can encode a protein sequence $x$ with the learned constraints. Specifically, we first encode the i-th amino acid by concatenating its corresponding single-site constraint as well as the possible pairwise ones:
\begin{equation}
\begin{split}
&\boldsymbol{\varepsilon_i}(x_i) =[\varepsilon_i(x_i), \varepsilon_{i1}(x_i, a_{1\boldsymbol{\cdot}}), ..., \varepsilon_{iM}(x_i, a_{M\boldsymbol{\cdot}})]  \\
&\varepsilon_{ij}(x_i, a_{j\boldsymbol{\cdot}}) = [\varepsilon_{ij}(x_i, a_1), \varepsilon_{ij}(x_i, a_2), ..., \varepsilon_{ij}(x_i, a_{20})] 
\end{split}
\end{equation}
where $\varepsilon_{ij}(x_i, a_{j\boldsymbol{\cdot}})$ gathers the 20 amino acids for any position $j \ne i$. Then we map $\boldsymbol{\varepsilon_i}(x_i)$ to the amino-acid embedding space with trainable parameter $W_{\varepsilon}$, and add the mapped vector to the original amino-acid embedding $e(x_i)$ to get the final feature vector as our model input:
\begin{equation}
\begin{split}
&\hat{e}(x_i) = e(x_i) + W_{\varepsilon} * \boldsymbol{\varepsilon_i}(x_i) \\
&H_0 = z', \quad H_{i} = \hat{e}(x_{i-1}) \; \text{for} \; 1 \le i < M
\end{split}
\label{eq:enhance-h}
\end{equation}
where $H_i$ is the input embedding of the Transformer decoder, i.e., the first  input is set to the sampled latent vector $z'$ and the input for other position $i$ is set to the combinatorial structure augmented feature vector  $\hat{e}(x_{i-1})$.

Combining Eq.~\eqref{eq:e-step-init} and \eqref{eq:enhance-h}, the learning process  becomes:\\
\textbf{E-step}:
\begin{align}
\tilde{\mathcal{L}_t} &= \sum_{n=1}^{N} w'(x_n, z_n) \Bigl(\log Q_\phi (x_n | z'_n; \boldsymbol{\varepsilon})  \\
& - \bigl( \frac{1}{2}\| \sigma_\psi(x_n) \|_2^2 + \frac{1}{2}\| \mu_\psi(x_n) \|^2_2 - \sum_{i=1}^d\log \sigma_{\psi,i}(x_n)   \bigr)\Bigr)  \notag\\
&w'(x_n, z_n)=\frac{ w(x_n, z_n)}{\sum_{n=1}^N w(x_n, z_n)}\notag \\
&w(x_n, z_n)=\frac{P_{\theta}(x_n|z_n;\boldsymbol{\varepsilon})}{Q_{\phi^{(t)}}(x_n|z_n;\boldsymbol{\varepsilon})}P(\mathcal{S}|x_n) \notag
\label{expectation}
\end{align}
where $Q_\phi (x_n | z'_n; \boldsymbol{\varepsilon})$  is  computed  from  
$\mathrm{Transformer}(H_0, H_{1:M})$. In practice, we also learn $P_{\theta}(x,z;\boldsymbol{\varepsilon})=P_0(z)P_{\theta}(x|z;\boldsymbol{\varepsilon})$ as a combinatorial structure feature enhanced latent generative model, where $P_{\theta}(x|z;\boldsymbol{\varepsilon})$ is a combinatorial structure feature enhanced Transformer with a learnable mapping matrix $W_{\eta}$ to transform the combinatorial structure features into the embedding space.
$d$ is the dimensionality of $\sigma_\psi(x_n)$ and $\sigma_{\psi, i}(x_n)$ is its $i$-th dimension.\\
\textbf{M-step}: \\
\begin{equation}
\begin{split}
\psi^{(t+1)} = \psi^{(t)} + r * \nabla_{\psi} \tilde{\mathcal{L}_t}, \phi^{(t+1)} = \phi^{(t)} + r * \nabla_{\phi} \tilde{\mathcal{L}_t} 
\end{split}
\label{maximization}
\end{equation}
The overall learning algorithm is given in Appendix~\ref{training_algorithm}.


\section{Experiments}
\label{experiments}
In this section, we conduct extensive experiments to validate the effectiveness of our proposed \model on protein sequence design task.

\subsection{Implementation Details} 
We first train a VAE model on raw protein sequences using a 6-layer Transformer as the encoder and a 2-layer Transformer as the decoder. 
We  add the combinatorial structure features to the Transformer decoder input as described in Sec.~\ref{subsection_mrfs} and its weight $W_{\eta}$ is learned jointly.   
The raw protein probability $P_{\theta}(x|z;\boldsymbol{\varepsilon})$ is then defined by the combinatorial structure feature augmented decoder in this VAE. 
The protein combinatorial structure constraints $\boldsymbol{\varepsilon}$ are learned on the training sequences for each dataset, rather than using real multiple sequence alignments (MSAs), to ensure a fair comparison.
The approximate posterior probability $\tilde{P}_\theta(z|x)$ is defined by its encoder. 
Its embedding and feed-forward network sizes are 320 and 1280. 
The encoder is initialized with the pre-trained ESM-2 weights~\cite{lin2022language}\footnote{\scriptsize {https://dl.fbaipublicfiles.com/fairesm/models/esm2\_t6\_8M\_UR50D.pt}}. 
The latent variable $z$'s dimension is also 320. 
We use another 6-layer Transformer encoder followed by a linear mapping to learn $\mu_{\phi}$ and $\sigma_{\phi}$ for $R_\psi(z|x)=\mathcal{N}(\mu_{\psi}, \sigma^2_{\psi}\bold{I})$ and another 2-layer Transformer decoder to learn $Q_\phi(x|z;\boldsymbol{\varepsilon})$. We initialize parameters $\psi^{(0)}$ and $\phi^{(0)}$ with $\theta$.


The number of iterations in the importance sampling-based MCEM is set to 10. The mini-batch size and learning rate are set to $4,096$ tokens and $1$e-$5$ respectively.
The model is trained with $1$ NVIDIA RTX A$6000$ GPU card.
We apply Adam algorithm~\cite{kingma2014adam} as the optimizer with a linear warm-up over the first $4,000$ steps and linear decay for later steps.
We randomly split each dataset into training/validation sets with the ratio of $9$:$1$. We run all the experiments for five times and report the average scores.
More experimental settings are given in Appendix \ref{addtional_experimental_setting}. Following \cite{kamisetty2013assessing}, we set $\alpha_{\text{single}}=1$ and $\alpha_{\text{pair}}=0.2*(M-1)$ with $M$ equals to the protein sequence length.

In inference, we design protein sequences by taking the wild-type as encoder input and 
the latent vector is sampled from prior distribution $Q_0(z)=N(0, \bold{I})$.
The sequences are decoded using sampling strategy with top-$5$. The candidate number is set to K=$128$ following the setting of \citet{jain2022biological} on the GFP dataset.

\subsection{Datasets}
Following \citet{ren2022proximal}, we evaluate our method on the following eight protein engineering benchmarks:\\
(1) \textbf{Green Fluorescent Protein (avGFP)}: The goal is to design sequences with higher log-fluorescence intensity values. We collect data following \citet{sarkisyan2016local}.
(2) \textbf{Adeno-Associated Viruses (AAV)}: The target is to generate amino-acid segment (position $561-588$) for higher gene therapeutic efficiency. We collect data following \citet{bryant2021deep}.
(3) \textbf{TEM-1 $\beta$-Lactamase (TEM)}: The goal is to design high thermodynamic-stable sequences. We merge the data from \citet{firnberg2014comprehensive}. 
(4) \textbf{Ubiquitination Factor Ube4b (E4B)}: The objective is to design sequences with higher enzyme activity. We gather data following \citet{starita2013activity}.
(5) \textbf{Aliphatic Amide Hydrolase (AMIE)}: The goal is to produce amidase sequences with higher enzyme activity. We merge data following \citet{wrenbeck2017single}.
(6) \textbf{Levoglucosan Kinase (LGK)}: The target is to optimize LGK protein sequences with improved enzyme activity. We collect data following \citet{klesmith2015comprehensive}. 
(7) \textbf{Poly(A)-binding Protein (Pab1)}: The goal is to design sequences with higher binding fitness to multiple adenosine monophosphates. We gather data following \citet{melamed2013deep}.
(8) \textbf{SUMO E2 Conjugase (UBE2I)}: We aim to find human SUMO E2 conjugase with higher  growth rescue rate. Data are obtained following \citet{weile2017framework}.
The detailed data statistics, including protein sequence length, data size and data source are provided in Appendix \ref{append_data}.

\begin{table*}[ht]
\small
\centering
\begin{tabular}{lcccccccc>{\columncolor{mygray}}c}
\midrule
Models  & avGFP & AAV & TEM & E4B & AMIE & LGK & Pab1 & UBE2I & \cellcolor{white}Average \\
\midrule
CMA-ES & $4.492$ & $-3.417$ & $0.375$ & $-0.768$ & $-8.224$ & $-0.077$ & $0.164$ & $2.461$ & $-0.624$\\
FBGAN & $1.251$  & $-4.227$  &  $0.006$ & $0.369$  &  $-2.410$ &  $-1.206$ & $0.029$ &    $0.208$    & $-0.747$   \\
DbAS & $3.548$ & $4.327$ & $0.003$ &$-1.286$ &$-2.658$ &$-1.148$ &$1.524$ & $3.088$ & $0.924$\\
CbAS & $3.550$ & $4.336$ &$0.106$ &$-1.000$ &$-1.306$ &$-0.362$ &$1.842$ &$3.263$ & $1.303$ \\
PEX & $3.764$ & $3.265$ &$0.121$ &$5.019$ &$-0.474$ &$0.007$ &$1.153$ &$1.995$ & $1.856$ \\
GFlowNet-AL & $5.062$ & $1.205$ &$1.552$ &$3.155$ &$0.059$ & $0.027$ &$2.168$ &$3.576$ & $2.101$\\
ESM-Search & $2.610$ & $-5.099$ & $0.148$ & $-1.860$ & $-2.351$ & $-0.029$ & $1.406$ & $3.244$ & $-0.241$ \\
\hdashline
\model & $\textbf{6.185}$ & $\textbf{4.813}$ & $\textbf{1.850}$ & $\textbf{5.737}$ & $\textbf{0.062}$ & $\textbf{0.035}$  & $\textbf{2.923}$ & $\textbf{4.536}$ & $\textbf{3.267}$\\
~-- w/o ESM & $1.214$ &  $-4.313$ &$0.005$ & $-1.352$  &$-6.376$ &$-0.225$ &$0.072$ &$1.843$ & $-1.141$ \\
~-- w/o ISEM & $4.708$ &$1.130$ &$0.708$ & $0.046$ & $-2.335$ &$-0.077$ &$1.913$ &$0.475$ & $1.342$ \\
~-- w/o MRFs & $4.376$ & $1.008$ &$0.952$ & $0.045$ &  $-1.771$ & $-0.012$  &$1.652$ &$2.418$ & $1.083$\\
~-- w/o LV & $4.274$ & $2.251$ &$0.078$ &$-1.612$ & $-2.266$ &$-0.931$ &$0.041$ &$-0.262$ & $0.196$\\
\bottomrule
\end{tabular}
\caption{Maximum fitness scores (MFS) of all methods on eight datasets. Higher values indicate better functional properties in the dataset. Our proposed \model achieves the highest fitness scores on all datasets.}
\label{table_fitness}
\end{table*}

\begin{table*}[ht]
\small
\centering
\begin{tabular}{lcccccccc>{\columncolor{mygray}}c}
\midrule
Models  & avGFP & AAV & TEM & E4B & AMIE & LGK & Pab1 & UBE2I & \cellcolor{white}Average \\
\midrule
CMA-ES & $\textbf{225.12}$ & $23.50$ & $261.60$ & $86.81$ & $283.90$ & $317.08$ & $61.16$ & $140.92$ & $175.01$ \\
FBGAN & $0.64$ &  $8.31$ &  $0.46$ & $3.87$ & $33.87$ & $17.35$ & $3.07$ &  $3.00$ & $8.82$ \\
DbAS &$3.04$ & $3.00$ &$3.67$ &$5.94$ &$1.32$ &$2.30$ & $4.05$ &$11.80$ & $4.33$\\
CbAS & $1.31$ &$3.01$ &$7.03$ &$7.09$ &$6.01$ &$6.15$ &$9.86$ &$22.73$ & $8.23$\\
PEX & $6.83$ &$4.35$ &$10.26$ &$5.22$ &$7.56$ &$13.24$ &$5.33$ &$10.32$ & $7.88$\\
GFlowNet-AL & $224.78$ &\textbf{25.57} &\textbf{266.43} &$43.62$ &$219.84$ &$212.25$ &$37.13$ & $49.79$ & $134.92$\\
ESM-Search & $3.79$ & $3.58$ & $11.56$ & $3.82$ & $3.83$ & $3.78$ & $5.71$ & $6.59$& $5.33$ \\
\hdashline
\model & $218.62$ & $22.92$ & $202.09$ & $\textbf{91.35}$ & $\textbf{293.30}$ & $\textbf{405.99}$ & $\textbf{68.27}$ & $122.66$ & $\textbf{178.15}$\\
~-- w/o ESM & $204.21$ & $13.87$ &$194.78$ &$7.90$ &$276.88$ &$362.98$ &$3.30$ &$119.87$ & $147.96$\\
~-- w/o ISEM & $122.15$ & $22.91$ &$70.12$ & $86.17$ &$145.74$ &$169.17$ &$60.22$ & $13.05$ & $153.09$ \\
~-- w/o MRFs & $217.35$ &  $17.02$ &$225.88$ &$84.64$ &$268.07$ & $381.66$ & $66.26$ &$\textbf{143.29}$ & $175.52$\\
~-- w/o LV &  $22.70$ &$5.26$ &$10.12$ & $5.24$ & $8.65$ &$20.28$ & $0.67$ & $2.98$ & $9.48$\\
\bottomrule
\end{tabular}
\caption{Diversity scores of all models on eight datasets. Higher values indicate more diverse protein sequences. Our \model achieves the highest average diversity scores over the eight protein datasets.}
\label{table_diversity}
\end{table*}

\begin{table*}[!t]
\small
\centering
\begin{tabular}{lcccccccc>{\columncolor{mygray}}c}
\midrule
Models  & avGFP & AAV & TEM & E4B & AMIE & LGK & Pab1 & UBE2I & \cellcolor{white}Average \\
\midrule
CMA-ES & $221.55$ & $22.73$ & $269.25$ & $93.78$ & $256.07$ & $415.63$ & $59.35$ & $128.10$ & $183.30$ \\
FBGAN & $0.05$ & $2.76$ & $0.08$ & $0.63$ & $57.87$ & $39.36$ &  $0.75$ & $0.80$ & $12.43$  \\
DbAS & $1.01$ & $3.01$ &$1.47$ &$1.09$ &$1.12$ &$1.63$ &$1.64$ &$2.05$ & $1.66$\\
CbAS & $4.02$ &$3.03$ &$2.06$ &$1.90$ &$1.33$ &$1.09$ &$2.71$ &$2.95$ & $1.92$\\
PEX & $3.59$ &$1.88$ &$8.57$ &$4.08$ &$4.50$ & $10.53$ & $3.63$ &$10.24$ & $5.87$\\
GFlowNet-AL & $221.95$ &$22.83$ &$266.99$ &$86.78$ &$316.79$ & $412.58$ & $61.33$ &$143.28$ & $191.56$\\
ESM-Search & $1.50$ & $1.83$ & $7.32$ & $1.21$ & $0.92$ & $0.91$ & $5.46$ & $6.14$ & $3.16$\\
\hdashline
\model & $\textbf{226.31}$ & $\textbf{23.81}$  & $\textbf{270.27}$  &$\textbf{96.57}$  &$\textbf{332.23}$ & $\textbf{420.93}$  & $\textbf{70.09}$  & $\textbf{153.27}$ & $\textbf{199.18}$\\
~-- w/o ESM & $198.08$ &$9.25$ &$176.20$ & $3.79$ & $264.36$ &$340.62$ & $1.49$ &$110.53$ & $138.04$\\
~-- w/o ISEM & $195.85$ &$16.36$ &$244.05$ & $85.81$ &$306.22$ &$382.12$ &$53.01$ & $99.51$ & $180.40$\\
~-- w/o MRFs & $207.59$ & $9.53$ &$221.49$ & $90.81$ &$244.14$ &$327.14$ &$58.44$ & $129.75$ & $161.11$\\
~-- w/o LV & $16.67$ & $0.89$ &$4.39$ & $0.48$ &$3.98$ &$9.64$ &$0.21$ &$1.76$ & $4.75$\\
\bottomrule
\end{tabular}
\caption{Novelty scores of all models on eight datasets. Higher values indicate more novel protein sequences. Our \model achieves the highest novelty scores on all the datasets.}
\label{table_novelty}
\end{table*}

\begin{figure*}
\begin{minipage}[t]{0.25\linewidth}
\centering
\includegraphics[width=4.5cm]{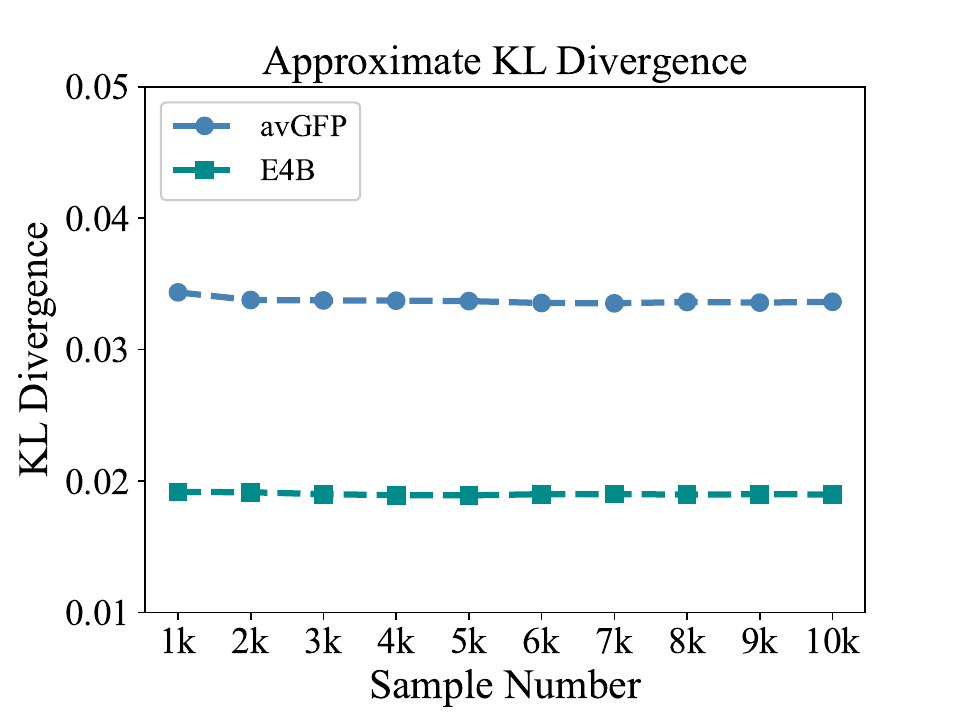}
\centerline{\small{(a) avGFP and E4B}}
\end{minipage}%
\begin{minipage}[t]{0.25\linewidth}
\centering
\includegraphics[width=4.5cm]{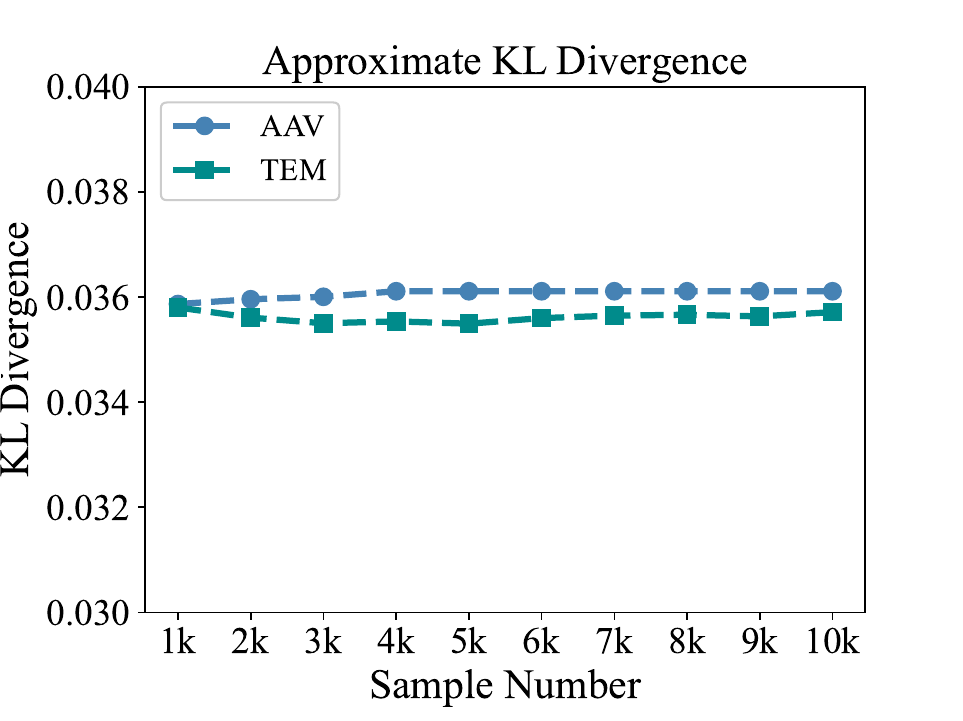}
\centerline{\small{(b) AAV and TEM}}
\end{minipage}%
\begin{minipage}[t]{0.25\linewidth}
\centering
\includegraphics[width=4.5cm]{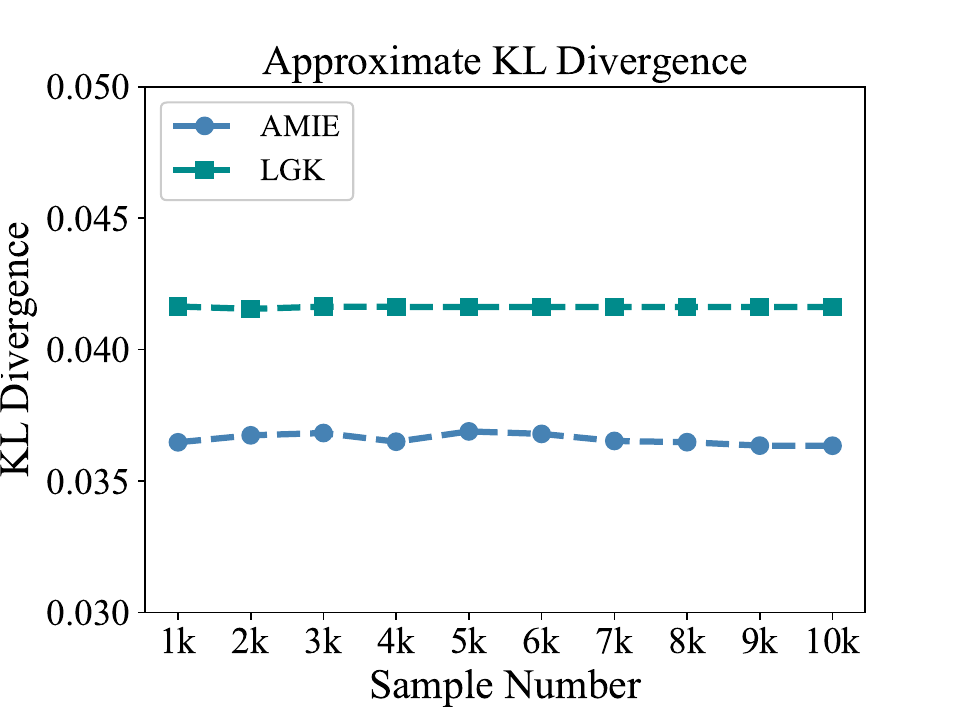}
\centerline{\small{(c) AMIE and LGK}}
\end{minipage}%
\begin{minipage}[t]{0.25\linewidth}
\centering
\includegraphics[width=4.5cm]{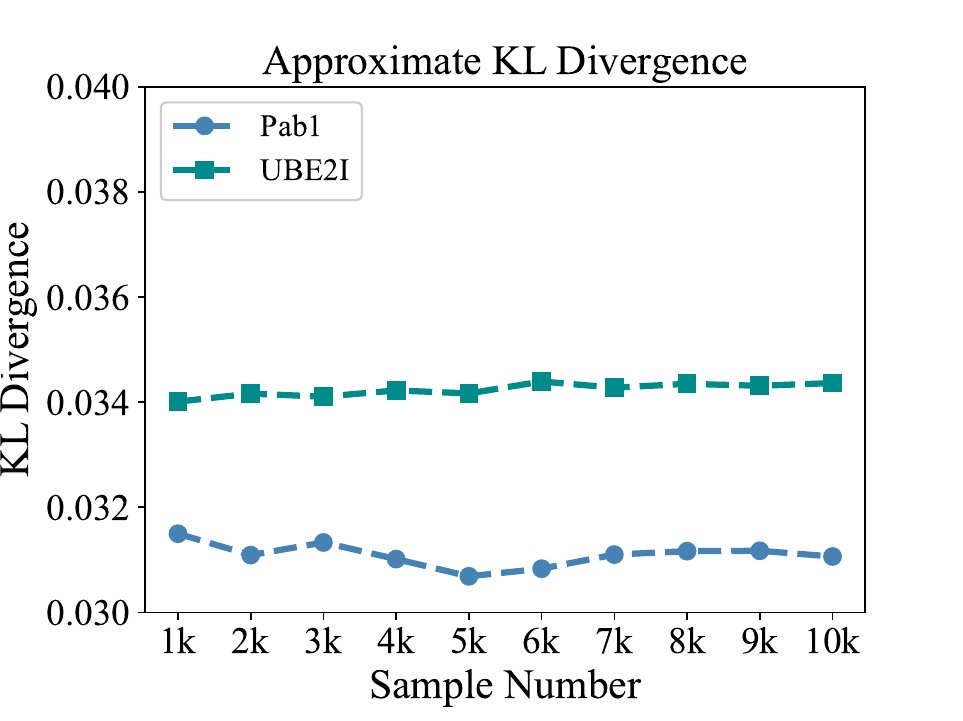}
\centerline{\small{(d) Pab1 and UBE2I}}
\end{minipage}
\vspace{-1em}
\caption{Approximate KL divergence on eight protein datasets. It shows the variance of KL divergence is very small over different sample size for all datasets, giving empirical evidence that when we sample from the ultimate $Q_{\phi}(x)$, it has minor difference compared with sampling from the posterior distribution $P_{\theta}(x|\mathcal{S})$.} 
 \label{figure_other_kl}
\end{figure*} 

\subsection{Baseline Models}
We compare our method against the following representative baselines:
(1) \textbf{CMA-ES} \cite{hansen2001completely} is a famous evolutionary search algorithm.
(2) \textbf{FBGAN} proposed by \citet{gupta2019feedback} is a novel feedback-loop architecture with generative model GAN.
(3) \textbf{DbAS} \cite{brookes2018design} is a probabilistic modeling framework and uses adaptive sampling algorithm.
(4) \textbf{CbAS} \cite{brookes2019conditioning} improves on DbAS by conditioning on the desired properties.
(5) \textbf{PEX} proposed by \citet{ren2022proximal} is a model-guided sequence design algorithm using proximal exploration.
(6) \textbf{GFlowNet-AL} \cite{jain2022biological} applies GFlowNet to design biological sequences. We use the implementations of CMA-ES, DbAS and CbAS provided in \citet{trabucco2022design} and for other baselines, we apply their released codes. 
To better analyze the influence of different components in our model, we also conduct ablation tests as follows:
(1) \textbf{\model-w/o-ESM} removes ESM-2 as encoder initialization.
(2) \textbf{\model-w/o-ISEM} removes iterative optimization process.
(3) \textbf{\model-w/o-MRFs} removes MRFs features and iterative optimization process.
(4) \textbf{\model-w/o-LV} removes latent variable, MRFs features and iterative optimization process.
(5) \textbf{ESM-Search} samples sequences from the softmax distribution obtained by finetuning ESM-2 on the protein datasets and taking the wild-type as input.

\subsection{Evaluation Metrics}
We use three automatic metrics to evaluate the performance of the designed sequences:
(1) \textbf{MFS}: Maximum fitness score. The oracle model adopted to evaluate MFS is described in Appendix \ref{addtional_experimental_setting}; (2) \textbf{Diversity} proposed by \citet{jain2022biological} is used to evaluate how different the designed candidates are from each other; (3) \textbf{Novelty} proposed by \citet{jain2022biological} is used to evaluate how different the proposed candidates are from the sequences in training data.

\subsection{Main Results}
Table \ref{table_fitness}, \ref{table_diversity} and \ref{table_novelty} respectively report the maximum fitness scores, diversity scores and novelty scores of all models.

\textbf{\model achieves the highest fitness scores on all protein families and outperforms the previous best method GFlowNet-AL by 55\% on average} (Table \ref{table_fitness}).
The reasons are two-folds.
On one hand, the importance sampling based MCEM can help our model to navigate to a better region instead of getting trapped in a worse local optima.
On the other hand, the combinatorial structure features help to recognize the preferred mutation patterns which have higher success rate under the nature selection pressure, potentially leading to sequences with higher fitness scores.

\textbf{\model achieves the highest average diversity score over the eight tasks} (Table \ref{table_diversity}).
Our model gains the highest diversity on $4$ out of $8$ tasks while GFlowNet-AL gains $2$ and CMA-ES gains $1$. It indicates that though involving combinatorial structure constraints can give the guidance for preferred protein patterns, it might also limit the sequence design to these patterns to some extent.
The involved latent variable can capture complex inter-dependencies among amino acids, which benefits for more diverse protein design. 

\textbf{\model can design more novel protein sequences on all datasets} (Table \ref{table_novelty}).
Our model achieves higher novelty scores on all datasets due to the reason that more new samples are involved during the importance sampling based iterative optimization process, which is beneficial for more novel protein design.

\subsection{Ablation Study}
Bottom halves of Table \ref{table_fitness}, \ref{table_diversity} and \ref{table_novelty} report the results of ablation tests.
\model-w/o-MRFs improves the average diversity and novelty scores by as much as $33$x compared with \model-w/o-LV, which demonstrates that introducing a latent variable can significantly help to generate diverse proteins.
\model-w/o-MRFs achieves higher maximum fitness scores than \model-w/o-ESM on all datasets, validating that adopting a pretrained protein language model as the encoder helps to design more satisfactory protein sequences.
However, directly finetuning ESM-2 to sample candidates (ESM-Search) drops $1.3$ points on average fitness score compared with taking ESM-2 as an encoder (\model-w/o-MRFs), demonstrating that ESM-2 is not suitable for direct sequence design.
Incorporating combinatorial structure features can further improve the fitness of the designed proteins (\model-w/o-ISEM V.S. \model-w/o-MRFs), based on which learning the proposal distribution by importance sampling based MCEM can better promote more desirable, diverse and novel protein generation.
We also provide the model performance with a fully-connected decoder~(non-autoregressive decoder) instead of the current autoregressive one in Appendix~\ref{appendix_nat}.
It shows the non-autoregressive decoder can not generate good candidates for protein with longer sequences such as LGK ($439$ amino acids).

\section{Analysis}
\label{analysis}


\begin{figure*}
\begin{minipage}[t]{0.5\linewidth}
\centering
\includegraphics[width=7.0cm]{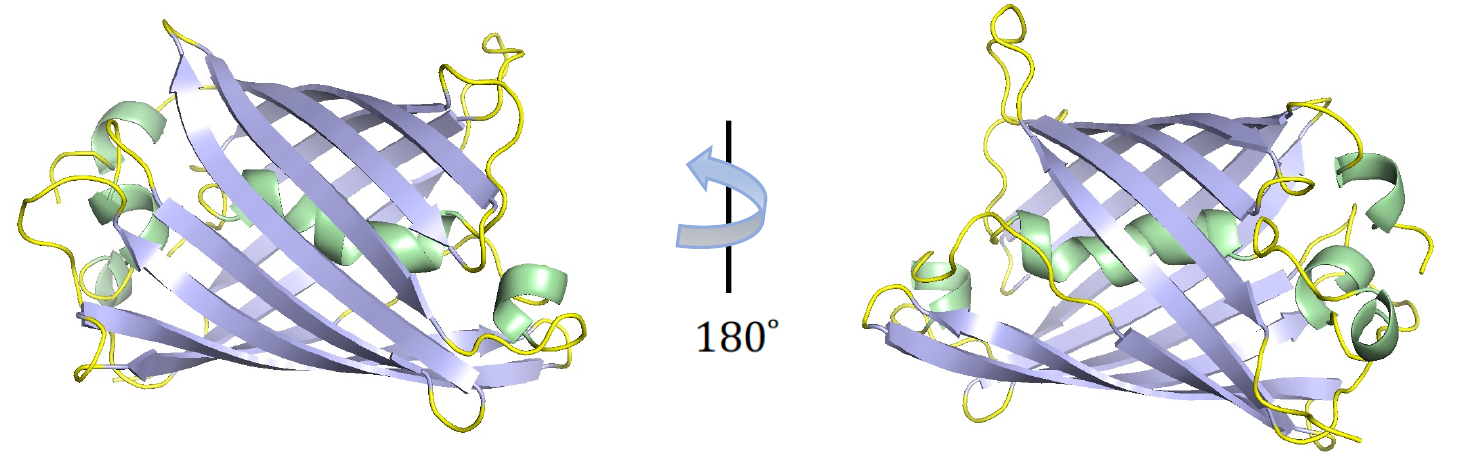}
\centerline{(a) $3$-D visualization of sequence with highest fitness.}
\end{minipage}%
\begin{minipage}[t]{0.5\linewidth}
\centering
\includegraphics[width=7.0cm]{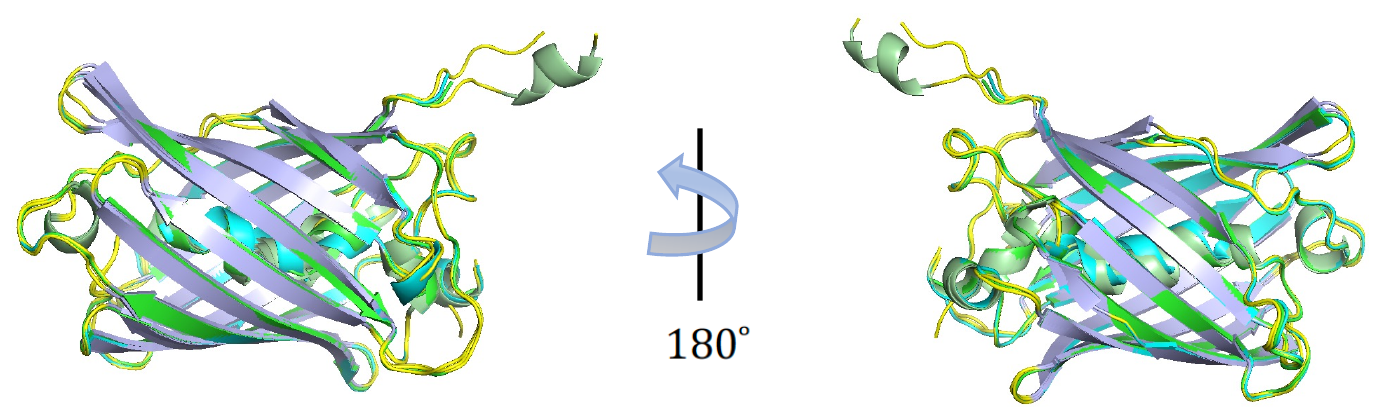}
\centerline{(b) Superposition of $5$ fluorescent proteins in top-5 templates.}
\end{minipage}
\vspace{-1.2em}
	\caption{3-D visualization of our designed  green fluorescent protein, validating that \model can generate realistic fluorescent protein.} 
 \label{figure_3d}
\end{figure*}




\begin{table}[!t]
\small
\centering
\begin{tabular}{lccc}
\midrule
Methods  & MFS & Diversity & Novelty  \\
\midrule
Latent-Add & $4.102$ & $218.26$ & $209.85$\\
Latent-Memory & $4.040$ & $\textbf{219.00}$ & $211.38$ \\
\model & $\textbf{6.185}$ & $218.62$ & $\textbf{226.31}$\\
\bottomrule
\end{tabular}
\caption{Results of different schemes of introducing a latent variable with a pretrained encoder evaluated on avGFP dataset. It shows that our method, which takes the latent representation as the first token input of decoder, achieves a higher fitness and novelty scores though with a mild decrease on diversity.}
\label{table_vae}
\end{table}

\subsection{Approximate KL Divergence}
To validate how close the proposal distribution $Q_{\phi}(x)$ is to the posterior distribution $P_{\theta}(x|\mathcal{S})$, we calculate the KL divergence through Monte Carlo approximation.
Here we calculate the KL divergence between $Q_{\phi}(x)$ and $P_{\theta}(x|\mathcal{S})$ as we have proven in Lemma $C.1$ (provided in Appendix \ref{theoretical_understanding}) that the sampling difference between these two distributions can be bounded under this divergence. 
We leverage an unbiased and low-variance estimator (proof shown in Appendix \ref{kl_proof}) to approximate the KL divergence as follows:
\begin{equation}
\small 
D_{KL}(Q_{\phi}(x)||P_{\theta}(x|\mathcal{S}))=E_{Q_{\phi} (x)}[r(x)-1-\log r(x)] 
\end{equation}
where $r(x)=\frac{P_{\theta}(x|\mathcal{S})}{Q_{\phi}(x)}$.  
The approximate KL divergence on eight protein datasets over $1$k-$10$k samples are illustrated in Figure \ref{figure_other_kl}.
From the figure, we can see that the variance of KL divergence is very small over different sample size for all datasets.
Besides, the KL divergence finally arrives at a small value, such as $0.018$ for E4B and $0.033$ for avGFP. It gives empirical evidence that when we sample from the ultimate $Q_{\phi}(x)$, it has minor difference compared with sampling from the posterior distribution $P_{\theta}(x|\mathcal{S})$.

\subsection{Effect of VAE Implementation Method}
Next, we study the effect of different implementation schemes of involving a latent variable with a pretrained encoder.
Some works have tried adding the latent representation to the original embedding layer (Latent-Add) or using it as an additional memory (Latent-Memory) when adopting a pretrained language model as encoder \cite{li2020optimus}.
We also implement our model with these two schemes, and evaluate the model performance on avGFP dataset. 
Table \ref{table_vae} shows that our method, which takes the latent representation as the first token input of decoder, achieves a higher fitness and novelty scores though with a mild decrease on diversity.
We also provide the conditional VAE~(CVAE) performance in Appendix~\ref{appendix_cvae}, which shows decreased performance. It demonstrates that iterative sampling inside the EM algorithm with explicit fitness constraints can lead to better proteins.

\subsection{Case Study}
To gain an insight on how well the designed proteins are, we analyze the generated avGFP sequence with highest fitness in detail using Phyre2 tool~\cite{kelley2015phyre2}. 
Figure~\ref{figure_3d} (a) illustrates the generated variant can fold stably.
According to the software, the most similar protein is Cytochrome b562 integral fusion with enhanced green fluorescent protein (EGFP) \cite{edwards2008linking}.
There are 227 residues (96\% of the candidate sequence) have been modeled with 100.0\% confidence by using this protein as template.
Details are given in Appendix \ref{case_study}.
Figure~\ref{figure_3d}(b) visualizes the superposition of the top-5 most similar templates to our sequence in the protein data bank, which are all fluorescent proteins and show highly consistent structure in most regions, validating that our model can design a real fluorescent protein.

\section{Related Work}
\label{related_work}
\textbf{Machine Learning for Protein Fitness Landscape Prediction.} Machine learning has been increasingly used for modeling protein fitness landscape, which is crucial for protein engineering. Some work leverage co-evolution information from multiple sequence alignments to predict fitness scores \cite{kamisetty2013assessing,luo2021ecnet}.
\citet{melamed2013deep} propose to construct a deep latent generative model to capture higher-order mutations.
\citet{meier2021language} propose to use pretrained protein language models to enable zero-shot prediction. 
The learned protein landscape models can be used to replace the expensive wet-lab validation to screen enormous designed sequences \cite{rao2019evaluating,ren2022proximal}.

\textbf{Methods for Protein Sequence Design.}
Protein sequence design has been studied with a wide variety of methods, including traditional directed evolution \cite{arnold1998design,dalby2011strategy,packer2015methods,arnold2018directed} and machine learning methods.
The mainly used machine learning algorithms include reinforcement learning \cite{angermueller2019model,jain2022biological}, Bayesian optimization \cite{belanger2019biological,moss2020boss,terayama2021black}, search using adaptive evolution methods \cite{hansen2006cma,swersky2020amortized}, likelihood-free inference \cite{zhang2021unifying}, deep generative models \cite{brookes2018design,madani2020progen,kumar2020model,das2021accelerated,hoffman2022optimizing,melnyk2021benchmarking, ren2022proximal} and latent deep generative model~\cite{brookes2019conditioning}. 
Our approach has the same objective using Bayes rule and KL divergence (Eq.~\eqref{eq:posterior_fitted_protein}~\eqref{equation4}) as CbAS~\cite{brookes2019conditioning}. 
Different from CbAS, our approach includes two latent models -- one representing the raw protein sequences and the other representing "good" proteins. The solution is derived under the Monte Carlo EM framework. Secondly, our approach is essentially self-training since we use both real proteins and (importance) sampled high-quality ones in each iteration to train the generation model. Thirdly, we augment the generative model with combinatorial structure features learned from MRF. 
This comprehensive framework not only enables the generative model to explore optimal regions in either the Fujiyama landscape or the Badlands landscape \cite{kauffman1989nk}, but also significantly enhances protein diversity and novelty.


\section{Discussion}
\label{discussion}

We opt for MCEM over the standard VAE to learn the latent generative model due to our belief that the standard VAE has certain limitations.
The major limitation of the standard VAE is that it maximizes the likelihood of observed data, say protein sequence, but higher likelihood does not necessarily associate with higher fitness of a protein. Our IsEM-Pro tackles this by explicitly modelling fitness constraints and amino acid correlation in a protein sequence. As shown in Table~\ref{table_fitness}, \ref{table_diversity} and \ref{table_novelty}, the fitness scores of standard VAE (\model-w/o-MRFs) decrease a lot compared with our \model and the diversity and novelty scores also slightly decrease. Besides, as \citet{dieng2019reweighted} analyzed, VAE amortizes the cost of inference by using a recognition network to parameterize the variational family, which introduces an amortization gap and leads to approximate posteriors of reduced expressivity due to the problem known as posterior collapse. Instead, EM directly maximizes the log marginal likelihood of the data and each iteration in EM is guaranteed to increase the log marginal likelihood from the previous iteration~\cite{bishop2006pattern}. Therefore, using EM in the context of deep generative models could lead to better performance.
Additionally, standard VAE attempts to perform variational inference through a direct, discriminative mapping from data observations to approximate posterior parameters. Though generative models can adapt to accommodate sub-optimal approximate posteriors, it's likely limited to direct inference mapping, leading to being trapped in a worse local optima~\cite{marino2018iterative}. Instead, EM guarantees that the log marginal likelihood of the data keeps increasing in each iteration, helping our model climb much closer to the global optima.

One limitation of this work is, although our model has demonstrated promising results, the designed protein sequences have not undergone wet-lab testing and there might be some uncertainty correlated with the oracle model.
The results reported in our paper are averaged over five runs, which can accommodate some variance on this metric. Furthermore, all fitness data used in our paper are obtained from wet-lab experiments, ensuring that the fitness values of the training datasets are realistic. To reduce the cost of wet-lab validation and select good protein candidates, we evaluate the designed sequences using the oracle model trained on real protein sequences following~\cite{ren2022proximal,jain2022biological}, which may associate with some uncertainties. While there is currently no perfect automatic metric available, we believe that new methods should be encouraged in the field of computational biology. With time, we are confident that the field will continue to develop and mature.

\section{Conclusion}
\label{conclusion}
This paper proposes \model, a latent generative model for protein sequence design, which incorporates additional combinatorial structure features learned by MRFs.
We use importance weighted EM to learn the model, which can not only enhance design diversity and novelty, but also lead to protein sequences with higher fitness.
Experimental results on eight protein sequence design tasks show that our method outperforms several strong baselines on all metrics.

\section*{Acknowledgements}
The authors would like to thank the anonymous reviewers for their valuable comments. We are also grateful to Siqi Ouyang, Yujian Liu, Jingjing Xu, Danqing Wang, Antonis Antoniades, and 
Jennifer Listgarten for their great suggestions. This work is partially supported by UCSB Faculty Research Award.


\bibliography{icml2023}
\bibliographystyle{icml2023}

\newpage
\appendix
\onecolumn
\section{Data Statistics}
\label{append_data}
We provide the detailed data statistics in the following table, including protein sequence length, data size and data source.
We have checked and cleaned the data and make sure the data do not contain personally identifiable information or offensive content.

\begin{table*}[ht]
\footnotesize
\begin{center}
\begin{tabular}{lccc}
\midrule
Protein & Length & Size & Data Source \\
\midrule
avGFP & $237$ & $49,855$ & https://figshare.com/articles/dataset/Local\_fitness\_landscape\_of\_the\_green\_fluorescent\_protein \\
AAV & $28$ & $296,914$ & https://github.com/churchlab/Deep\_diversification\_AAV\\
TEM & $286$ & $17,238$ & https://github.com/facebookresearch/esm/tree/main/examples/data\\
E4B & $102$ & $91,033$ & https://figshare.com/articles/dataset\\
AMIE & $341$ & $6,631$ & https://figshare.com/articles/dataset/Normalized\_fitness\_values\_for\_AmiE\_selections/3505901/2\\
LGK & $439$ & $8,069$ & https://figshare.com/articles/dataset \\
Pab1 & $75$ & $36,522$ & https://figshare.com/articles/dataset \\
UBE2I & $159$ & $5,355$ & http://dalai.mshri.on.ca/~jweile/projects/dmsData/\\
\bottomrule
\end{tabular}
\end{center}
\caption{Detailed statistics of the eight protein datasets.}
\label{table_data_statistics}
\end{table*}

\section{More Implementation Details}
\label{experimental_setting}
\subsection{Additional Experimental Settings}
\label{addtional_experimental_setting}
We apply the annealing schedule for the KL term during $P_{\theta}(x,z)$ training process following $\beta$-VAE \cite{higgins2022beta} to prevent posterior collapse. Specifically, the KL term coefficient starts from $0$ and is gradually increased to $1.0$ as training goes on. 
At each iteration in importance sampling based EM learning process, the number of samples from current $Q_{\phi}(x, z)$ is set to $10$\% of the training data size.

Following \cite{ren2022proximal}, We construct the oracle model $f(x)$ by adopting the features produced by ESM-1b~\cite{rives2021biological} with dimension $1280$ and finetuning an Attention1D decoder to predict the fitness values.
Since \citet{brookes2019conditioning} state that the results are insensitive when $\lambda$ is set in the range [$50$, $100$]-th percentile of the fitness scores in the training set, we set $\lambda$ to $50$-th percentile of the fitness values in the training data to accommodate more diversity. 

\subsection{KL Divergence for Two Gaussian Distribution}
\label{kl_divergence}
The kl divergence for two Gaussian distribution is defined as follows:
\begin{equation}
\begin{split}
p(x)&= \frac{1}{2\pi^{0.5n}|\Sigma|^{0.5}} \exp{(-\frac{1}{2}(x-\mu)^T\Sigma^{-1}(x-u))} \\
q(x)&= \frac{1}{2\pi^{0.5n}|L|^{0.5}} \exp{(-\frac{1}{2}(x-m)^TL^{-1}(x-m))} \\
D_{KL}(p||q)&=\frac{1}{2}\{\log\frac{|L|}{|\Sigma|} + Tr(L^{-1}\Sigma)+(\mu-m)L^{-1}(u-m)^T-n\} 
\end{split}
\label{eq:normal-KL}
\end{equation}
When $p(x)=\mathcal{N}(\mu, \sigma^2 \boldsymbol{I})$ and $q(x)=\mathcal{N}(0, \boldsymbol{I})$, the KL divergence can be computed in closed form as follows:
\begin{equation}
D_{KL}(p||q)=\frac{1}{2}\| \sigma \|_2^2 + \frac{1}{2}\| \mu \|^2_2 - \sum_{i=1}^d \log \sigma_i-\frac{d}{2}
\label{equation: special_kl}
\end{equation}
where $d$ is the dimensionality of $\sigma$.

\subsection{Full Derivation Details about EM Algorithm}
\label{em_detail}
We provide the derivation details about EM algorithm in Eq (7) and Eq (8) as follows:\\
E-step:
\begin{align}
&\text{Sample}~~ x_n \sim P_{data}, ~~ z_n\sim R_{\psi}(z|x_n) 
\text{ and } z_n\sim Q_0(z), ~~ x_n\sim Q_{\phi}(x|z_n), z'_n \sim R_{\psi}(z|x_n)  \\
\mathcal{L}_t &= \frac{\sum_{n=1}^N w(x_n, z_n) \mathcal{F}(R_{\psi}(z|x), \phi)}{\sum_{n=1}^N w(x_n, z_n)} \notag \\
&=\sum_{n=1}^N w'(x_n, z_n) \Big\{ E_{R_{\psi}(z|x)} [\log Q_{\phi}(x, z) - \log R_{\psi}(z|x)] + D_{KL}(R_{\psi}(z|x)||Q_{\phi}(z|x)) \Big\} \notag \\
&\ge \sum_{n=1}^N w'(x_n, z_n) \Big\{ E_{R_{\psi}(z|x)} [\log Q_{\phi}(x, z) - \log R_{\psi}(z|x)] \Big\} \notag \\
&= \sum_{n=1}^N w'(x_n, z_n) \Big\{ E_{R_{\psi}(z|x)} [\log Q_{\phi}(x|z) + \log Q_0(z) - \log R_{\psi}(z|x)] \Big\} \notag \\
&= \sum_{n=1}^N w'(x_n, z_n) \Big\{ E_{R_{\psi}(z|x_n)} [\log Q_{\phi}(x_n|z) + \log Q_0(z) - \log R_{\psi}(z|x_n)] \Big\} \notag \\
&=\sum_{n=1}^N w'(x_n, z_n) \Big\{ \log Q_\phi (x_n | z'_n) - D_{KL}(R_{\psi}(z|x_n)||Q_0(z)) \Big\} \notag \\
&=\sum_{n=1}^N w'(x_n, z_n) \Big\{ \log Q_\phi (x_n | z'_n) - \bigl( \frac{1}{2}\| \sigma_\psi(x_n) \|_2^2 + \frac{1}{2}\| \mu_\psi(x_n) \|^2_2 - \sum_{i=1}^d\log \sigma_{\psi,i}(x_n)   \bigr) \Big\} \notag \\
&=\tilde{\mathcal{L}_t}
\label{eq:e-step-full}
\end{align}
where $w(x_n, z_n)=\frac{P_{\theta}(x_n|z_n)}{Q_{\phi^{(t)}}(x_n|z_n)}P(\mathcal{S}|x_n)$ is the unnormalized importance weight, $w'(x_n, z_n)=\frac{ w(x_n, z_n)}{\sum_{n=1}^N w(x_n, z_n)}$ is the normalized one, and $N$ is the sample size.  
We use two techniques to generate samples for $z$ and $x$  -- using the real protein sequences $x$ with generated $z$ from $R_\psi(z|x)$ and using $z$ from a standard normal distribution with generated proteins from $Q_\phi(x|z)$. Since $R_\psi(z|x)$ is assumed to be a normal distribution with its mean $\mu_\psi$ and variance $\sigma_\psi^2 \bold{I}$ calculated from a Transformer encoder $\mu_{\psi}, \sigma_{\psi}=\mathrm{Transformer}_\psi(x_n)$, the KL term can be calculated in close form (Eq. \eqref{equation: special_kl}). $d$ is the dimensionality of $\sigma_\psi(x_n)$ and $\sigma_{\psi, i}(x_n)$ is the $i$-th dimension of $\sigma_\psi(x_n)$. $\tilde{\mathcal{L}_t}$ is the lower bound of $\mathcal{L}_t$.\\\\
M-step: 
\begin{equation}
\begin{split}
\psi^{(t+1)}, \phi^{(t+1)} &= \argmax_{\psi, \phi} \tilde{\mathcal{L}_t} \quad \Rightarrow \quad \psi^{(t+1)} = \psi^{(t)} + r * \nabla_{\psi} \tilde{\mathcal{L}_t}, \phi^{(t+1)} = \phi^{(t)} + r * \nabla_{\phi} \tilde{\mathcal{L}_t} \\
\nabla_{\psi} \tilde{\mathcal{L}_t} 
&= \sum_{n=1}^N w'(x_n, z_n) \Big\{ - \nabla_{\psi} \bigl( \frac{1}{2}\| \sigma_\psi(x_n) \|_2^2 + \frac{1}{2}\| \mu_\psi(x_n) \|^2_2 - \sum_{i=1}^d\log \sigma_{\psi,i}(x_n)   \bigr) \Big\} \\
\nabla_{\phi} \tilde{\mathcal{L}_t} &= \sum_{n=1}^N w'(x_n, z_n) \Big\{ \nabla_{\phi} \log Q_\phi (x_n | z'_n) \Big\}
\end{split}
\end{equation}  
$r$ is the learning rate. $\nabla_{\psi} \mathcal{L}_t$ and $\nabla_{\phi} \mathcal{L}_t$ are the gradients of $\psi$ and $\phi$ at the $t$-th iteration, respectively. 

\subsection{IsEM-Pro Algorithm}
\label{training_algorithm}
\begin{algorithm}[ht] 
\caption{IsEM-Pro Model Learning} 
\label{alg::training} 
\begin{algorithmic}[1]
\REQUIRE
$\boldsymbol{\varepsilon}$: separately learned combinatorial structure features through MRFs\\
$P_{\theta}(x,z; \boldsymbol{\varepsilon})$: a pre-trained protein sequence probability model (VAE's decoder augmented with $\boldsymbol{\varepsilon}$)\\
T: number of iteration for importance sampling based EM learning\\
$\mathcal{X}_{\text{data}}$: real protein sequences\\
$N_{data}$: number of real protein sequences\\
$r$: learning rate \\
$\mathcal{B}$: batch size \\
$d$: dimensionality of the latent variable
\ENSURE 
Final proposal model $Q_{\phi^{(T)}}(x|z; \boldsymbol{\varepsilon})$ and variational distribution $R_{\psi^{(T)}}(z|x)$
\STATE set $Q_{\phi^{(0)}}(x|z; \boldsymbol{\varepsilon})=P_{\theta}(x|z; \boldsymbol{\varepsilon})$, \; $R_{\psi^{(0)}}(z|x)=\tilde{P}_\theta(z|x)$
\FOR{t=0 to T-1}
\STATE $D_1=\emptyset$, $D_2=\emptyset$
\FOR{each $x_n\sim \mathcal{X}_{\text{data}}$}
\STATE sample $z_n\sim R_{\psi^{(t)}}(z|x_n)=\mathcal{N(\mu_{\psi}, \sigma^{\text{2}}_{\psi} \boldsymbol{I})}\}$
\STATE $D_1=D_1\cup \{(x_n, z_n)\}$
\ENDFOR
\FOR{i=1 to ($N_{data}*0.1)$}
\STATE sample $z_i\sim \mathcal{N}(0, \boldsymbol{I})$
\STATE sample $x_i\sim Q_{\phi^{(t)}}(x|z_i;\boldsymbol{\varepsilon})$
\STATE $D_2=D_2\cup \{(x_i,z_i)\}$
\ENDFOR
\STATE $D=D_1\cup D_2$, $N=|D|$
\FOR{minibatch $\{(x_n, z_n)\}_{n=1}^\mathcal{B}\subset D$}
\STATE $w(x_n, z_n)=\frac{P_{\theta}(x_n|z_n;\boldsymbol{\varepsilon})}{Q_{\phi^{(t)}}(x_n|z_n;\boldsymbol{\varepsilon})}P(\mathcal{S}|x_n)$, $w'(x_n, z_n)=\frac{ w(x_n, z_n)}{\sum_{n=1}^\mathcal{B} w(x_n, z_n)}$
\STATE sample $z'_n\sim R_{\psi^{(t)}}(z|x_n)=\mathcal{N}(\mu_{\psi}, \sigma^2_{\psi} \boldsymbol{I})$
\STATE $\tilde{\mathcal{L}_t} = \sum_{n=1}^{\mathcal{B}} w'(x_n, z_n) \Big\{ \log Q_{\phi^{(t)}} (x_n | z'_n;\boldsymbol{\varepsilon}) - \bigl( \frac{1}{2}\| \sigma_{\psi^{(t)}}(x_n) \|_2^2 + \frac{1}{2}\| \mu_{\psi^{(t)}}(x_n) \|^2_2 - \sum_{i=1}^d\log \sigma_{\psi^{(t)},i}(x_n)   \bigr) \Big\}$
\STATE update $\psi^{(t)} = \psi^{(t)} + r * \nabla_{\psi} \tilde{\mathcal{L}_t}$
\STATE update $\phi^{(t)} = \phi^{(t)} + r * \nabla_{\phi} \tilde{\mathcal{L}_t}$
\ENDFOR
\STATE update $\psi^{(t+1)} = \psi^{(t)}$
\STATE update $\phi^{(t+1)} = \phi^{(t)}$
\ENDFOR
\end{algorithmic}
\end{algorithm}
\newpage
\section{Approximate KL Divergence} 
\subsection{Proof of the Unbiased and Low-Variance Estimator}
\label{kl_proof}
Letting $r(x)=\frac{P_{\theta}(x|\mathcal{S})}{Q_{\phi} (x)}$, we have:
\begin{equation}
\begin{split}
E_{Q_{\phi} (x)}[(r(x)-1)-\log r(x)] = E_{Q_{\phi}(x)}[\log \frac{Q_{\phi}(x)}{P_{\theta}(x|\mathcal{S})}] 
= D_{KL}(Q_{\phi}(x)||P_{\theta}(x|\mathcal{S}))
\end{split}
\end{equation}
Therefore, this estimator for KL divergence is unbiased.

Let $y=r(x)$ and $f(y)= (y-1) - \log y$, since $f(y)$ is a convex function and it achieves the minimum value when $y=1$, we have:
\begin{equation}
(y-1) - \log y \ge f(1) = 0 
\end{equation}

Thus, $(r(x)-1)-\log r(x)$ is always larger than or equals to $0$.
Instead, in the original KL divergence, $\log \frac{Q_{\phi} (x)}{P_{\theta}(x|\mathcal{S})} = -\log r(x)$  would be negative for half of the samples. Therefore, $E_{Q_{\phi} (x)}[(r(x)-1)-\log r(x)]$ has lower variance compared to the original one.



\subsection{Theoretical Understanding}
\label{theoretical_understanding}
We can prove that under acceptable KL divergence, the samples from the proposal distribution $Q_{\phi}(x)$ can be bounded within a reasonable sampling error with samples from the posterior distribution $P_{\theta}(x|\mathcal{S})$.

\begin{lemma}
If the KL divergence between two distributions P and Q is less than a small positive value $\delta$, then the sampling probability difference between P and Q will be bounded by $\sqrt{2\delta}$ for each sample.
\end{lemma}
\begin{proof}\renewcommand{\qedsymbol}{}
Let $\delta(P_{\theta}(x|\mathcal{S}), Q_{\phi}(x))$ be the total variation distance between $P_{\theta}(x|\mathcal{S})$ and $Q_{\phi}(x)$. We have:
\begin{equation}
\frac{1}{2}\sum_{x} |P_\theta(x|\mathcal{S})-Q_\phi(x)| = \delta(P_{\theta}(x|\mathcal{S}), Q_{\phi}(x)) \le \sqrt{\frac{1}{2} D_{KL}(P_{\theta}(x|\mathcal{S})||Q_{\phi}(x))} < \sqrt{\frac{1}{2}\delta}
\label{pinsker}
\end{equation}
\end{proof}
The first inequality is due to Pinsker's inequality \cite{csiszar2011information}, which is tight if and only if $P_\theta(x| \mathcal{S})=Q_\phi(x)$. Then there is no difference between sampling from $P_\theta(x| \mathcal{S})$ and $Q_\phi(x)$.

From the above analysis, we can get:
\begin{equation}
\sum_{x} |P_\theta(x|\mathcal{S})-Q_\phi(x)| < \sqrt{2\delta}
\end{equation}
When $\delta$ approaches $0$, the sampling difference between $P_\theta(x|\mathcal{S})$ and $Q_\phi(x)$ would be very minor.

\section{Case Study}
\label{case_study}
Figure \ref{figure_ss} illustrates the complete sequence and secondary structure analyse of our designed protein of avGFP compared with Cytochrome b562 integral fusion with enhanced green fluorescent protein (EGFP). From the figure, we can see that there are much overlap between our designed protein sequence and the chain B of Cytochrome b562 integral fusion with EGFP. It gives empirical evidence that the green fluorescent protein generated by our model is highly likely to be a real protein compared with the proteins we already know. But whether the designed sequences can accelerate wet-lab experiments  still need more exploration as we can not 100\% trust it.

\begin{figure}[ht]
  \centering
  \includegraphics[width=17.0cm]{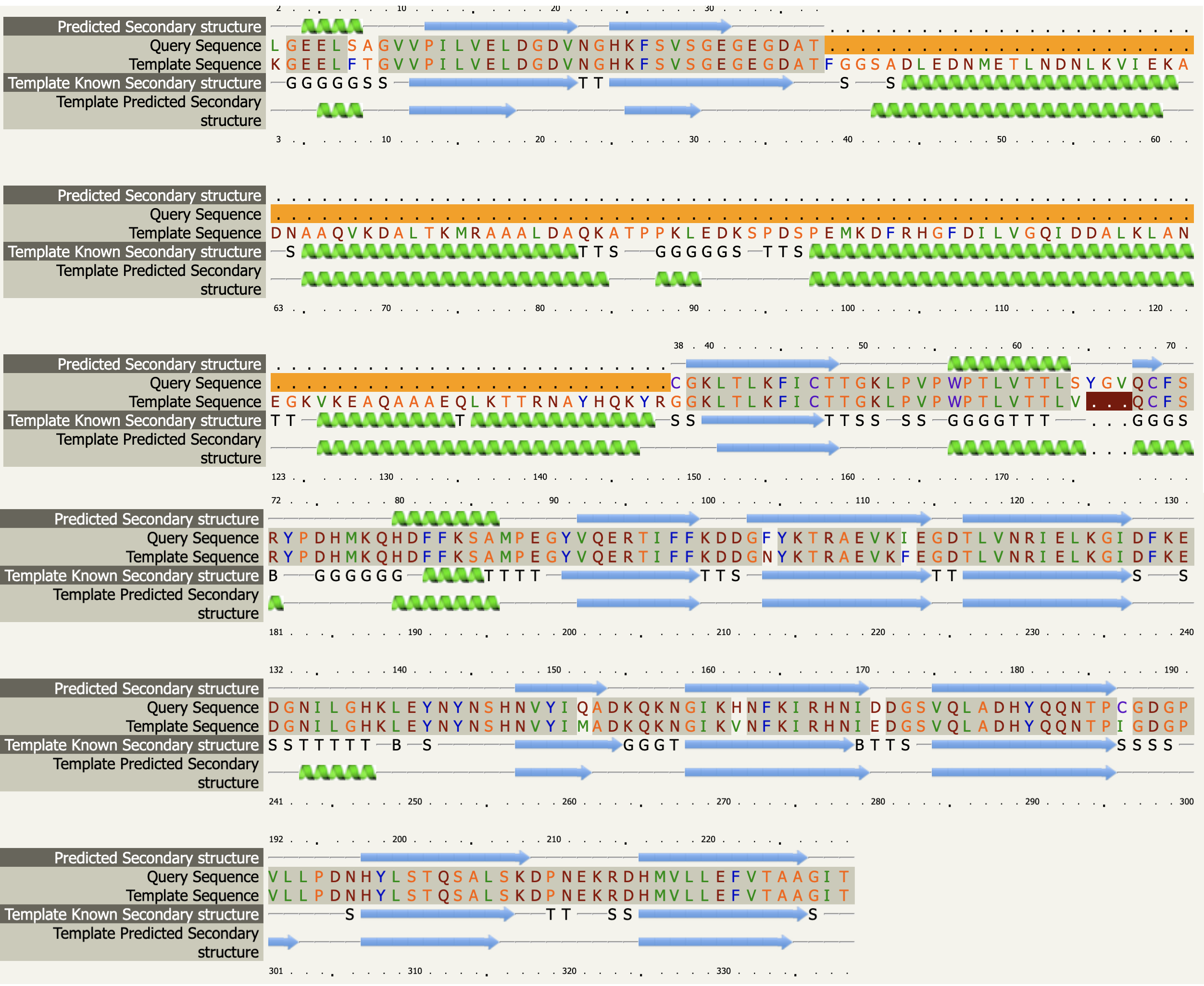}
  \caption{Complete sequence and secondary structure comparison between the designed sequence of avGFP and chain B of Cytochrome b562 integral fusion with with enhanced green fluorescent protein. Green and blue parts respectively represent the Alpha helix and Beta strand.}
  \label{figure_ss}
\end{figure}

\section{Pseudo-Likelihood for Combinatorial Structure Learning}
\label{pseudo_likelihood}
We train the Markov random fields using a pseudo-likelihood as CCMpred \cite{seemayer2014ccmpred} additionally combined with the L1 regularization and L2 regularization. The pseudo-likelihood is given in the following equation:
\begin{equation}
\begin{split}
\hat{P}_{\epsilon}(x) &= \log \Pi_{i=1}^M P_{\epsilon}(x_i|x_1, x_2, ..., x_{i-1}, x_{i+1}, ..., x_M, \boldsymbol{\varepsilon})\\
&= \sum_{i=1}^M \log \frac{\exp(\varepsilon_i(x_i) + \sum_{j=1, j\ne i}^M\varepsilon_{ij}(x_i, x_j))}{\sum_{c\in \mathcal{V}}\exp(\varepsilon_i(c) + \sum_{j=1, j\ne i}^M\varepsilon_{ij}(c, x_j))} \\
&= \sum_{i=1}^M \{\varepsilon_i(x_i) + \sum_{j=1, j\ne i}^M\varepsilon_{ij}(x_i, x_j) - \log Z_i\} \\
Z_i &= \sum_{c\in \mathcal{V}}\exp(\varepsilon_i(c) + \sum_{j=1, j\ne i}^M\varepsilon_{ij}(c, x_j))
\end{split}
\end{equation}
where $\mathcal{V}$ denotes the vocabulary of $20$ amino acids.

\section{Additional Experiments}
\subsection{Conditional VAE}
\label{appendix_cvae}

We provide the results of conditional VAE~(CVAE) in Table~\ref{table_cvae}. We implement the CVAE which is built by adding the fitness condition on deepsequence~\cite{riesselman2018deep}. 
The table shows CVAE exhibits decreased performance. It demonstrates that iterative sampling inside EM algorithm with explicit fitness constraints can lead to better proteins.

\begin{table*}[ht]
\begin{center}
\begin{tabular}{lccccccccc}
\midrule
Dataset & avGFP & AAV & TEM & E4B & AMIE & LGK & Pab1 & Average \\
\midrule
CVAE & $3.570$ &	$4.128$	& $0.083$ &	$-0.819$ &	$-1.564$ &	$-0.796$ &	$1.638$ &	$3.492$	& $1.217$ \\
\model & $6.185$ &	$4.813$ &	$1.850$ &	$5.737$ &	$0.062$ &	$0.035$ &	$2.923$ &	$4.536$	& $\textbf{3.267}$ \\ 
\bottomrule
\end{tabular}
\end{center}
\caption{Fitness of conditional VAE compared with our \model.}
\label{table_cvae}
\end{table*}

\subsection{Non-Autoregressive Decoder}
\label{appendix_nat}
We present a performance comparison between our \model and a fully-connected non-autoregressive decoder in Table~\ref{table_nat}. The non-autoregressive model, a VAE, includes a 6-layer Transformer encoder followed by a linear mapping to learn the mean and variance of the latent variable, and a two-layer non-autoregressive Transformer decoder. The embedding and feed-forward network dimensions are set to 320 and 1280, respectively. The model uses sampled latent vector as decoder input for all positions and is trained using the glancing strategy~\cite{qian2021glancing}. The results show that the non-autoregressive model struggles to generate high-quality candidates for proteins with longer sequences, such as LGK, which has 439 amino acids.

\begin{table*}[ht]
\begin{center}
\begin{tabular}{lccccccccc}
\midrule
Dataset & avGFP & AAV & TEM & E4B & AMIE & LGK & Pab1 & Average \\
\midrule
Non-autoregressive Decoder & $1.312$ &	$3.137$ &	$1.015$	& $3.096$ &	$-3.297$ &	$-3.125$ &	$1.548$	& $2.694$ &	$0.797$ \\
\model & $6.185$ &	$4.813$ &	$1.850$ &	$5.737$ &	$0.062$ &	$0.035$ &	$2.923$ &	$4.536$	& $\textbf{3.267}$ \\ 
\bottomrule
\end{tabular}
\end{center}
\caption{Fitness of a non-autoregressive decoder based VAE model compared with our \model.}
\label{table_nat}
\end{table*}

\end{document}